\documentclass[12pt]{article}
\pdfoutput=1

\usepackage[latin1]{inputenc}
\usepackage[english]{babel}
\usepackage{graphicx}
\usepackage{natbib}
\usepackage{amsfonts,amssymb,amsmath,amscd,amsthm,latexsym}
\usepackage{mathrsfs}

\usepackage[a4paper,citecolor=blue,linkcolor=blue,colorlinks=true]{hyperref}
\usepackage{appendix}
\usepackage{algorithm,algpseudocode}
\usepackage{enumerate}
\usepackage{txfonts}

\newcommand{\var}{\operatorname{Var}}

\newcommand{\eps}{\varepsilon}
\renewcommand{\epsilon}{\varepsilon}
\newcommand{\dd}{\mathrm{d}}

\newcommand{\prob}{\operatorname{\mathbb P}}

\newcommand{\esp}{\mathbb{E}}

\newcommand{\R}{\mathbb{R}}

\newcommand{\piamis}{\widehat{\Pi}^\text{AMIS}}
\newcommand{\piast}{\widehat{\Pi}^\ast}

\newtheorem{theo}{Theorem}
\newtheorem{lem}[theo]{Lemma}
\newtheorem{pro}[theo]{Proposition}
 
\newtheorem{defi}[theo]{Definition} 

\theoremstyle{definition}

\title{Consistency of the Adaptive Multiple Importance Sampling}

\author{
Jean-Michel Marin$^{1}$, Pierre Pudlo$^{1,\, 2}$ and
Mohammed Sedki$^{1,\, 3}$   \\
{\small $^1$ Universit\'e Montpellier 2, I3M UMR CNRS 5149,
  France }\\
  {\small $^2$ INRA, CBGP UMR 1062, Montferrier-sur-Lez, France.}\\
  {\small $^3$ Université Paris-Sud, CESP Inserm U1018, Paris, France}
}

\date{May 23$^\text{th}$, 2014}
\usepackage{lineno}

\begin{document}

\maketitle

\begin{abstract}
  Among Monte Carlo techniques, the importance sampling requires fine
  tuning of a proposal distribution, which is now fluently resolved
  through iterative schemes. 
  The Adaptive Multiple Importance Sampling (AMIS) of \cite{cornuet:2012}
  provides a significant improvement in stability and effective sample size
  due to the introduction of a recycling
  procedure. 
  However, the consistency of the AMIS estimator remains largely open.
  In this work we prove the convergence of the AMIS, at a cost of a
  slight modification in the learning process.  Contrary to
  \cite{douc:2007:a}, results are obtained here in the asymptotic
  regime where the number of iterations is going to infinity while the
  number of drawings per iteration is a fixed, but growing sequence of
  integers. Hence some of the
  results shed new light on adaptive population Monte
  Carlo algorithms in that last regime. 
\end{abstract}



\section{Introduction}
The aim of Monte Carlo techniques is to approximate a target distribution
$\Pi(\cdot)$ on some space $\mathscr X$ with a weighted sample. Namely, they
output a system of particles $X_e\in\mathscr X$ (indexed by $e\in E$), with
their weights $\omega_e\in [0;\infty)$. Then, the discrete measure
\[
\widehat\Pi_E(\cdot)  = \underset{e \in E}{\sum} \omega_e \delta_{X_e}(\cdot)
\]
serves as an approximation of the target $\Pi(\cdot)$. And the Monte
Carlo scheme is said to be consistent if, for a large class of
functions $\psi:\mathscr X\to\mathbb R$, the sum $\int \psi(x)
\widehat\Pi_E(\text{d}x)=\sum_e \omega_e \psi(x_e)$ tends to the
integral $\Pi(\psi)=\int \psi(x) \Pi(\text{d} x)$ when the sample
size (\textit{i.e.}, the cardinality of the index set $E$) tends to
infinity.  Let us assume that the target $\Pi$ has a density
$\pi(\cdot)$ with respect to a reference measure $\mathrm{d}x$.
Among the Monte Carlo methods, Importance Sampling
\citep[see][]{hesterberg:1988, hesterberg:1995, ripley} consists in
drawing $X_e$ from a proposal probability $Q$ with density
$q(\cdot)$ and then in weighting this draw with $\omega_e =
\pi(X_e) / q(X_e)$ to take into account the discordance between
the sampling distribution and the target. The precision of the
importance sampling approximation becomes very unsatisfactory when the
importance weight $\omega_e$ has high variance, see
\citet{owen:2000}. Such problematic cases arise when the proposal
distribution $Q$ is not well adapted to the target $\Pi$ or when the
space $\mathscr X$ is of high dimensionality. Practically, the fine
tuning of the proposal measure for a given target is
difficult. However, there are adaptive techniques to learn the
proposal probability sequentially \citep[see][]{liu:book,
  rubinstein:book, pennanen:chapter}.  Among these adaptive methods,
the PMC scheme of \citet{cappe:2004}, generalised in the D-kernel
paradigm \citep[see ][]{douc:2007:a, douc:2007:b, cappe:2008}, aims at
building a mixture distribution by minimising some optimality
criterion (Kullback-Leibler, variance,\ldots).  Given a parametric
family of distributions $Q(\theta)$ indexed by $\theta \in \Theta$,
\citet{cappe:2008} seek the optimal proposal $Q(\theta^\ast)$ for a
given target $\Pi$ by estimating $\theta^\ast$ sequentially on
successive samples.

\bigskip

In many real problems where computing $\pi(X_e)$ (hence the importance
weight) is time consuming, recycling the successive samples generated during
the learning process is essential.  To this end, \cite{cornuet:2012} introduce
the Adaptive Multiple Importance Sampling (AMIS), combining multiple
importance sampling methods and adaptive techniques.  The AMIS is a
sequential scheme in the same vein as \cite{cappe:2008}. During the
learning process, the AMIS tries successive proposal
distributions, say
$Q\big(\widehat{\theta}_1\big),
\ldots,Q\big(\widehat{\theta}_T\big)$.
Each stage of the iterative process estimates a better proposal
$Q\big(\widehat{\theta}_{t+1}\big)$, by minimising a criterion such
as, for instance, the Kullback-Leibler divergence between $Q(\theta)$ and
$\Pi$.  But the novelty of the AMIS is the following recycling
procedure of all past simulations.  At iteration $t$, the AMIS has
already produced 
$t$ samples:\par
\begin{align*}
X_1^1,\ldots, X_{N_1}^1 & \sim Q\big(\widehat{\theta}_1\big),
\\
X_1^2,\ldots, X_{N_2}^2 & \sim Q\big(\widehat{\theta}_2\big),
\\
& \vdots
\\
X_1^t,\ldots, X_{N_t}^t & \sim Q\big(\widehat{\theta}_t\big)
\end{align*}
with respective sizes $N_1,N_2,\ldots, N_t$. Then the scheme derives a
new parameter $\widehat{\theta}_{t+1}$ from all those past
simulations. To that purpose, the weight of $X_i^k$ $(k \le t, i \le
N_k)$ is updated with
\begin{equation}
\label{eq:weight:t}
  \pi\big(X_i^k\big) \Big/\left[\sum_{\ell = 1}^{t} 
\frac{N_\ell}{\Omega_{t}} q\big( X_i^k,
  \widehat{\theta}_\ell\big)\right],
\end{equation}
where $\widehat{\theta}_1, \ldots, \widehat{\theta}_t$ are the
parameters generated throughout the $t$ past iterations, $x\mapsto
q(x,\theta)$ is the density of $Q(\theta)$ with respect to the
reference measure $\dd x$ and $\Omega_t = N_1 + N_2 + \cdots + N_t$
is the total number of past particles.  The importance weight
\eqref{eq:weight:t} is inspired by the techniques
of 
\cite{veach:1995}. \citet {owen:2000} popularise those techniques to
merge several independent importance samples and even propose a more
refined and stabilising alternative, named \textit{deterministic
  multiple mixture}.  On various numerical experiments where the
target is the posterior distribution of some population genetics data
sets, \cite{cornuet:2012} and \cite{siren:2010} show considerable
improvements of the AMIS in Effective Sampling Size \citep[denoted
further ESS, see][chapter~2]{liu:book}.  In such settings where
calculating $\pi(X_i^k)$ is drastically time consuming, a
recycling process makes sense. However, no proof of
convergence had yet been provided in \cite{cornuet:2012}. It is worth
noting that the weight \eqref{eq:weight:t} introduces long memory
dependence between the samples, and even a bias which was not
controlled by theoretical results. The main purpose of this paper is
to fill in this gap, and to prove the consistency of the algorithm at
the cost of a slight modification in the adaptive process. We suggest
learning the new parameter $\widehat{\theta}_{t+1}$ on the last
sample $X^t_1, \ldots, X^t_{N_t}$ weighted with the classical
$\pi\big(X^t_i\big) \Big/ q\Big(X^t_i, \widehat{\theta}_t\Big)$
for all $i = 1,\ldots, N_t$.  The only recycling procedure is in the
final output that merges all the previously generated samples using
\eqref{eq:weight:t}.

\bigskip

In \cite{douc:2007:a} for instance, the consistency of the adaptive
population Monte Carlo schemes is proven assuming that the number of
iterations, say $T$, is fixed and that the number of simulations
within each iteration, $N=N_1=N_2=\cdots= N_T$, goes to infinity.  We
decided to adopt a more realistic asymptotic setting in this
paper. Contrary to these last results, the convergence of
Theorem~\ref{thm:Consistency} holds when $N_1,\ldots, N_T$ is a
growing, but fixed sequence and $T$ goes to infinity. Hence the proofs
of Theorem~\ref{thm:strong} provide new insights on adaptive PMC in
that last asymptotic regime.  The convergence of $\widehat{\theta}_t$
to the target $\theta^\ast$ relies on a weak law of large numbers on
triangular arrays \citep[see][Chapter~9]{chopin:2004, douc:2008,
  cappe:book}, and a clever application of the Chebyshev inequality to
obtain the almost sure consistency. 
The consistency of the final merging with weights given by
\eqref{eq:weight:t} is not a straightforward consequence of asymptotic
theorems. Its proof requires the introduction of a new weighting 
\begin{equation}
  \label{eq:wstar}
  \pi(X_i^k)\bigg/ q(X_i^k,\theta^\ast)
\end{equation}
that is more simple to study, although biased and non explicitly
computable (because $\theta^\ast$ is unknown). Under the set of
assumptions given below, this last weighting scheme is consistent (see
Proposition~\ref{pro:auxiliary}) and is comparable to the actual weighting
given by \eqref{eq:weight:t}, which yields the consistency proven in
Theorem~\ref{thm:Consistency}.

\bigskip

The paper is organised as follows. The modification of the original
AMIS is detailed in Section~\ref{sec:algorithms}. The main results
are in Section~\ref{sec:results}.  The proofs are in Sections
\ref{sec:proofs} and \ref{sec:proofs2}. At last, in
Section~\ref{sec:numerical}, we compare the performance of our new
algorithm against the original AMIS and a scheme with a naive recycling
strategy. 

\section{Modifications of the AMIS scheme}
\label{sec:algorithms}

\begin{algorithm}[h]
  \caption{Modified AMIS}
  \begin{algorithmic}[1]
    \Require an initial parameter
    $\widehat{\theta}_1$ and increasing sample sizes $N_1, \ldots, N_T$.
    \For {$t = 1 \to T$} \label{line:begin.learn}
    \For {$i = 1 \to N_t$}
    \State \textbf{draw} $X^t_i$ from $Q (\widehat{\theta}_t)$
    \State \textbf{compute} \label{line:simple.weight}
    $\omega^t_i = \pi(X_i^t)\big/ q(X_i^t, \widehat \theta_t)$.
    \EndFor
    \State \textbf{compute} \label{line:new.theta}
    $\widehat \theta_{t+1} = N_t^{-1}\sum_{i=1}^{N_t} \omega_i^t h(X_i^t)$
    \EndFor \label{line:end.learn}
    \State \textbf{set} $\Omega_T =N_1+\cdots + N_T$
    \For {$t = 1 \to T$} \label{line:begin.recycle}
    \For {$i = 1 \to N_t$}
    \State \textbf{update} \label{line:updated.weight}
    $\omega_i^t= \pi(X_i^t)\Big/ \Omega_T^{-1}\sum_{k=1}^T N_k
    q(X_i^t, \widehat\theta_k)$
    \EndFor
    \EndFor \label{line:end.recycle}
    \State \Return 
    $ \big(X^1_1,\omega^1_1\big), \ldots, \big(X^1_{N_1},
      \omega^1_{N_1}\big), \ldots, \big(X^T_1,
      \omega^T_1\big), \ldots, \big(X^T_{N_T}, \omega_{N_T}^T\big)$
  \end{algorithmic}
\end{algorithm}

When compared with the recursive algorithm of \citet{cappe:2008}, the
novelty in the AMIS of \cite{cornuet:2012} lies in the recycling
process at each iteration and in the final system which includes all
particles generated during the $T$ iterations of the adaptive
process. Hence, the AMIS estimation of the integral $\Pi(\psi) = \int
\psi(x) \Pi(\dd x)$ is
\begin{equation}
\label{eq:amis:estimator}
\widehat{\Pi}_T^{\text{AMIS}}(\psi) = \frac{1}{\Omega_T} \sum_{t = 1}^T
\sum_{i = 1}^{N_t} \left[ \frac{\pi(X_i^t)}{\Omega_T^{-1}\sum_{k=1}^T N_k q(X_i^t, \widehat\theta_k)} \right] \psi(X_i^t),  
\end{equation}
based on the weights given in~\eqref{eq:weight:t}.

The scheme we propose relies on the same ideas: we fit the proposals
assuming that the optimal proposal is at $\theta=\theta^\ast$,
whatever the criterion we believe in (Kullback-Leibler divergence,
variance criteria, moment fits,\ldots). The only condition is that we
are able to write the optimal parameter as $\theta^\ast=\int
h(x)\Pi(\dd x)$, where $h$ is an explicitly known function, so that
\begin{equation}
\widehat\theta_{t+1}=\frac{1}{N_t}\sum_{i=1}^{N_t}
\frac{\pi(X^t_i)}{q(X^t_i,\widehat\theta_t)} h(X^t_i).\label{eq:thetat}
\end{equation}
Our
process ends with the recycling of all past particles by updating all
weights with \eqref{eq:weight:t}.  But, contrary to
\cite{cornuet:2012}, the calibration of the new parameter
$\theta_{t+1}$ only considers the current sample $X_1^t, \ldots,
X^t_{N_t}$.  We hope improvements in the accuracy of
$\theta_{t+1}$ against previous estimations by requiring that the
sample size $N_t$ grows at each iteration. As a consequence, the
influence of the first estimations $\widehat{\theta}_1,\ldots,
\widehat{\theta}_{t-1}$ decreases significantly in the mixture of the
denominator of the importance weight \eqref{eq:weight:t}.\\

The modified AMIS is given in Algorithm~1. The learning process
between lines~\ref{line:begin.learn} and \ref{line:end.learn} draws a
sequence of samples from which it calibrates gradually the parameter
$\theta$. The new value of the proposal's parameter we compute at
line~\ref{line:new.theta} depends only on the last sample we have
drawn. This is the only discrepancy from the original algorithm of
\citet{cornuet:2012}.  Here, the recycling process producing the final
output is silently done by updating the weights of all produced
samples between lines~\ref{line:begin.recycle} and
\ref{line:end.recycle}. Finally, we should note that, if calculating
$\pi(x)$ is time consuming, the value computed at
line~\ref{line:simple.weight} should be stored in memory to perform
the update at line~\ref{line:updated.weight} during the recycling
process.

\section{Consistency results}
\label{sec:results}

We state here our main results (on the learning process in
Paragraph~\ref{sub:learning}, and on the final output in
Paragraph~\ref{sub:final}). For the sake of clarity, their proofs are
postponed to Sections~\ref{sec:proofs} and \ref{sec:proofs2}. We begin
with some hypothesis on the parametric family of proposals.

\subsection{Assumptions on the proposals}

We assume that the space $\Theta$ is a subset of the space $\R^d$
endowed with the Euclidian norm $\|\cdot\|$. The set $\mathscr X$ is a
subset of a finite-dimensional vector space, equiped with a reference
measure $\dd x$. All $Q(\theta)$ for $\theta\in \Theta$ and $\Pi$ are
absolutely continuous with respect to the reference measure. They have
densities $q(x,\theta)$ and $\pi(x)$ respectively. The minimal
hypothesis for importance sampling schemes to provide consistent
estimates is that $\Pi$ is absolutely continuous with respect to all
proposals: $\forall \theta\in \Theta$, $\Pi \ll Q(\theta)$. Hence we
assume that $q(x,\theta)=0$ implies $\pi(x)=0$.

\bigskip

Without loss of generality, we might assume that $\mathscr X$ and
$\Theta$ are both open subsets of Euclidian spaces, and that both
$\pi(x)$ and $q(x,\theta)$ are positive for all $x\in\mathscr X$,
$\theta\in \Theta$. Besides we can note that $\widehat\theta_{t+1}$ is
defined in \eqref{eq:thetat} as a linear combination of (random)
values of $h$, and the only fact we can safely affirm on the
coefficients of this combination is that they are positive. Therefore,
for the algorithm not to stop, any positive linear combination of
elements of $\Theta$ should fall into $\Theta$. In particular, this
implies that $\Theta$ cannot be a bounded subset of a Euclidian space.

\bigskip

We also impose some regularity conditions on the family of proposals
$\{Q(\theta)\}_{\theta \in \Theta}$ which will ensure consistency of
our procedure. For all $x \in \mathscr X$, $\theta \mapsto q(x,
\theta)$ is continuous on $\Theta$, and the joint function $(x,
\theta) \mapsto q(x, \theta)$ is lower semicontinuous on $\mathscr X
\times \Theta$. Moreover, when $\theta\to\theta^\ast$, $q(\cdot,
\theta)$ converges to $q(\cdot, \theta^\ast)$ uniformly over compact
sets, \textit{i.e.}, for any compact subset $K$ of $\mathscr X$,
\[
\| q(\cdot, \theta) - q(\cdot, \theta^\ast) \|_{K,\infty} := 
\sup\big\{ |q(x, \theta) - q(x, \theta^\ast)|\, :\ x\in K \big\}
\]
converges to $0$.

\subsection{Consistency of the learning process}
\label{sub:learning}
We focus here on the learnt $\widehat{\theta}_t$ defined in
\eqref{eq:thetat} and show convergence to the optimal parameter
$\theta^\ast=\int h(x)\pi(x)\dd x$. Sadly, the AMIS weight of a
particle, see \eqref{eq:amis:estimator}, is an average over the path
in the parameter space being taken during the sequential
algorithm. Weak consistency, \textit{i.e.}, convergence in
probability, is not enough to control such averages, as there exists
no Cesàro Lemma for the convergence in probability, see for instance
\citet{Billingsley}, exercise 20.23 p. 272. Thus, we decided to rely on
almost sure convergence. The challenge is to prove that the sequential
algorithm do not accumulate Monte Carlo errors over iterations.
Let us introduce the following class of functions. 
\begin{defi}\label{defi:G2}
  A function $\psi:\mathscr X\to\mathbb R^d$ belongs to the class
  $\mathscr G^2(\mathbb R^d)$ if and only if $\int
  \pi^2(x)\|\psi(x)\|^2/q(x,\theta)\dd x $ is finite for all $\theta$
  in $\Theta$ and depends continuously on $\theta$.
\end{defi}
We can interpret the integrability condition in the above definition
as follow: the classical importance Monte Carlo algorithm that
estimate $\Pi(\psi)$ has finite variance whatever the importance
distribution in the parametric family.
With such assumptions on the function $h$ used to learn the optimal
parameter, and a condition on the sample sizes, we can show the
following result.
\begin{theo}
  \label{thm:strong}
  If $h\in \mathscr G^2(\mathbb R^d)$, the estimate
  $\widehat{\theta}_t$ tends to $\theta^\ast$ in probability when
  $t\to\infty$.  If, additionally, $\sum_t 1/N_t<\infty$, then
  $\widehat{\theta}_t\to \theta^\ast$ almost surely.
\end{theo}
The proof, which is in Section~\ref{sec:proofs}, consists of two parts. In the
first part, we prove the convergence in probability using the weak law of
large numbers on triangular arrays given in the Appendix.
Afterwards, since $h\in\mathscr G^2(\mathbb R^d)$, conditionally on
$\widehat\theta_t$, the next estimate $\widehat\theta_{t+1}$ is the
sum of an iid and square integrable sample of size $N_t$. Thus, using the
Chebyshev inequality artfully, and the weak law of large number,
we obtain the almost sure convergence.

\subsection{Consistency of the final recycling scheme}
\label{sub:final}
The remaining part of our results deals with the final output 
that merges all the samples. More precisely, Theorem~\ref{thm:Consistency}
below says that the empirical sum of a function $\psi$ on the merged, re-weighted
sample provides a consistent approximation of the integral $\Pi(\psi)$. The
class of integrands $\psi\in\mathbb L^1(\Pi)$ for which the above holds is
determined by the following class of functions.
\begin{defi}\label{defi:H2}
  A function $\psi:\mathscr X\to \mathbb R$ belongs to the class
  $\mathscr H^2(\mathbb R)$ if and only if the integral $ \int
  \big[\pi(x) \psi(x)\big/ q\big(x, \theta^\ast\big)\big]^2 q\big(x,
  \theta\big) \dd x $ is finite for all $\theta\in \Theta$ and is a
  function of $\theta$ that is continuous at
  $\theta=\theta^\ast$.
\end{defi}
Likewise, the above class of functions might be interpreted in terms
of quadratic moments. Note that, if $\psi$ is in $\mathscr H^2(\mathbb R)$, then $\psi$
is in $L^1(\Pi)$. And finally, set
\begin{equation} \label{eq:mesp}
m_\eps(x) := \inf\{q(x, \theta)\,:\ \theta\in \bar B(\theta^\ast,\eps) \}
\end{equation}
where the infimum is actually attained because $\theta\mapsto
q(x,\theta)$ is continuous, therefore positive on $\mathscr X$.  We are now
in a position to state the following strong consistency.
\begin{theo} 
  \label{thm:Consistency}
  Assume that $h\in \mathscr G^2(\mathbb R^d)$ and $\sum_t
  1/N_t<\infty$. Moreover, assume that, for some $\eps > 0$,
  $\psi(\cdot) q(\cdot,\theta^\ast)/m_{\eps}(\cdot)$ is in $\mathscr
  H^2(\mathbb R)$.  Then, the sum over the
  final weighted sample $\widehat{\Pi}_T^{\text{AMIS}}(\psi)$ given
  in \eqref{eq:amis:estimator} tends almost surely to $\int \psi(x)
  \pi(x) \dd x$ when $T\to\infty$.
\end{theo}
The function $q(\cdot,\theta^\ast)/m_{\eps}(\cdot)$ is larger than $1$
on $\mathscr X$, and goes to $1$ as $\eps\to 0$. Hence the assumption
that $\psi(\cdot) q(\cdot,\theta^\ast)/m_{\eps}(\cdot)$ is in $\mathscr
  H^2(\mathbb R)$ is a bit stronger than just $\psi$ is in $\mathscr
  H^2(\mathbb R)$.

\section{Proofs of the convergence during the learning process}
\label{sec:proofs}

This Section is devoted to the proof of Theorem
\ref{thm:strong}. We first collect useful lemmas on the class
$\mathscr G^2(\mathbb R^2)$, then rely on the weak law of large number
of triangular arrays to obtain the convergence in
probability. Finally, we prove the almost convergence of the learnt parameters.

\subsection{Technical results on the functions of class $\mathscr
  G^2(\mathbb R^d)$}

The following lemma deals with the continuity condition imposed in
$\mathscr G^2(\mathbb R^d)$. With this result, it becomes obvious
that, if some function $\psi$ belongs to $\mathscr G^2(\mathbb R^d)$,
and if $\varphi$ is another function such that
$\|\varphi(x)\|\le\|\psi(x)\|$ for all $x\in \mathscr X$, then
$\varphi$ is also in $\mathscr G^2(\mathbb R^d)$.
\begin{lem}\label{lem:G2}
  Assume that, for any $\theta \in \Theta$, the integral  
  \[
  v_\psi(\theta) := \int
  {\pi^2(x) \|\psi(x)\|^2}/{q\big(x,\theta\big)} \dd x
  \]
  is finite. Fix $\theta\in \Theta$. These conditions are equivalent:
 {\em (i)}  the function $v_\psi$ is continuous at $\theta$; and
  {\em(ii)}  when $\theta'\to\theta$, 
    \[
    \int \pi^2(x)\|\psi(x)\|^2\left|\frac
      1{q(x,\theta')}- \frac 1{q(x,\theta)} \right|\dd x \to 0.
    \]
\end{lem}
\begin{proof}
  Clearly, {(ii)} implies {(i)}. It remains to show that
  {(i)} implies {(ii)}, that is to say, if we assume
  {(i)}, then, for any
  sequence $\theta_n$ that converges to $\theta$,
  $\lim_n A_n = 0$, where
  \[
  A_n:=\int \pi^2(x)\|\psi(x)\|^2\left|\frac 1{q(x,\theta_n)}- \frac
    1{q(x,\theta)} \right|\dd x.
  \]
  To this aim, fix a random variable $X$ with distribution $\pi$ and
  set 
  \[
  Z_n:= \pi(X) \|\psi(X)\|^2 \Big/ q(X,\theta_n),\quad Z=\pi(X)
  \|\psi(X)\|^2 \Big/ q(X,\theta)
  \]
  so that $\esp(Z_n)=\esp|Z_n|=\int \pi^2(x) \|\psi(x)\|^2 \Big/
  q(x,\theta_n)\dd x$.
  With the continuity conditions on the family $Q(\theta)$, $Z_n\to Z$
  (almost) surely, and with {(i)},
  $\esp|Z_n|\to\esp|Z|$. Hence $Z_n$ is uniformly integrable and
  $A_n=\esp|Z_n-Z|$ tends to $0$.
\end{proof}


In order to apply Theorem~\ref{thm:Cappe:1}, we will also need the
uniform integrability on compact set written below.
\begin{lem}\label{lem:A1}
  If $\psi$ is in $\mathscr G^2(\mathbb R^d)$, on any compact set
  $K\subset\Theta$, we have
  \[
    \lim_{\eta \to +\infty}\  \sup_{\theta \in K}\ 
    \int \pi(x) \|\psi(x)\| \mathbf 1\left\{ 
      \frac{\pi(x) \|\psi(x) \|}{q\big(x, \theta\big) } > 
    \eta  \right\}\dd x = 0.
  \]
\end{lem}
\begin{proof} Fix a compact set $K$.
  Since $\mathbf 1\{|y|>\eta\}\le \eta^{-1}|y|$  for any $y\ge0$, 
  we have
  \[
  \int \pi(x) \|\psi(x)\| \mathbf 1\left\{ 
      \frac{\pi(x) \|\psi(x) \|}{q\big(x, \theta\big) } > 
    \eta  \right\}\dd x \le \eta^{-1} 
  \int \frac{\pi^2(x) \|\psi(x)\|^2}{q(x, \theta)}\dd x.
  \] 
  The last integral depends continuously on $\theta$ and,
  consequently, is bounded by some finite $M$ on the compact set
  $K$. Therefore
  \[
  \sup_{\theta\in K}  \int \pi(x) \|\psi(x)\| \mathbf 1\left\{ 
      \frac{\pi(x) \|\psi(x) \|}{q\big(x, \theta\big) } > 
    \eta  \right\}\dd x \le \frac{M}{\eta},
  \]
  and the desired result is proven.
\end{proof}

\subsection{Proof of weak consistency of Theorem~\ref{thm:strong}}
\label{sub:prove1}

We define the 
$\sigma$-fields $\mathcal F_t=\sigma \big(X_1^1, \ldots, X^1_{N_1}, \ldots,
X^{t-1}_1, \ldots,X^{t-1}_{N_{t-1}} \big)$ which form a filtration.
We will derive the convergence applying
Theorem~\ref{thm:Cappe:1} on the triangular array given by
\begin{equation*}
  V_{t,i} :=\frac{\pi\big(X_i^t\big)}{ N_tq\big(X_i^t,
    \widehat{\theta}_{t}\big)}h\big(X_i^t\big).
\end{equation*}
Indeed, we have $\widehat\theta_{t+1}=\sum_{i=1}^{N_t} V_{t,i}$
The main difficulty here is in checking assumption \textit{(iii)} of that
theorem. Below, we begin by proving that $\widehat\theta_t$ is tight, then we
check that assumption \textit{(iii)} of Theorem~\ref{thm:Cappe:1} is true, and
conclude our proof.

\paragraph{Tightness of $\widehat\theta_t$.} We have 
 \begin{equation*}
    \big\|\widehat \theta_{t+1} \big\| \leq \frac{1}{N_t} \sum_{i =1}^{N_t}
    \frac{\pi\big(X^t_i\big) }%
    {q\big(X^t_i, \widehat   \theta_{t}\big)} \big\|h\big(X^t_i\big) \big\|.
  \end{equation*}
Moreover 
  \begin{equation*}
    \esp\left(\frac{\pi(X^t_i) }{q\left(X^t_i, 
          \widehat   \theta_{t}\right)} \big\|h(X^t_i)\big\| \Bigg|  \mathcal F_t  \right) 
    = \int \pi(x) \big\|h(x)\big\| \dd x,
  \end{equation*}
 therefore $\esp \left( \big\| \widehat{\theta}_{t+1} \big\|\right) 
\le \Pi\left(\big\|h\big\|\right)$. 
Then Markov's inequality leads to 
 \begin{equation*}
  \underset{c \to +\infty}{\lim} \underset{t}{\sup} \prob \big(\| \widehat
      \theta_t \| > c \big) \le \underset{c \to +\infty}{\lim} \frac{\Pi\big(\|h\|\big)}{c}
      = 0.
    \end{equation*}
Hence, the sequence $\big\{\widehat{\theta}_t\}_{t \ge 1}$ is tight.

\paragraph{Checking condition \textit{(iii)} of Theorem~\ref{thm:Cappe:1}.}
We have to show that, for any $\eta>0$, the following $S_t$ tends to $0$ in
probability, with
\begin{align*}
  S_t&:=\sum_{i=1}^{N_t} \esp \Big[\big\|V_{t,i}\big\|\mathbf 1\Big\{
  \big\|V_{t,i} \big\|  > \eta \Big\} \Big| \mathcal F_t \Big] 
  = \int \pi(x)\big\|h(x)\big\| \mathbf 1
  \left\{\frac{\pi(x)\|h(x)\|}{q\big(x,
      \widehat{\theta}_{t}\big)} > N_t \eta \right\}\dd x \\
  & = I(N_t\eta, \widehat\theta_t) \quad \text{by setting }
  I\left(\eta,  \theta \right) := \int \pi(x) \|h(x)\| 
  \mathbf 1\left\{ \frac{\pi(x) \|h(x)\|}{ q(x, \theta) } > \eta \right\} \dd x.
\end{align*}

Fix any $\epsilon > 0$. Using tightness of $\widehat{\theta}_t$ proven above,
we might introduce a compact set $K_\epsilon$ such that, for all $t\ge 1$,
$\prob(\widehat\theta_t\not\in K_\epsilon)\le \epsilon$.  On the event
$\big\{\widehat\theta_t\not\in K_\epsilon\big\}$, we have $I(N_t\eta, \widehat\theta_t)\le
I(0, \widehat{\theta}_t)=\pi(\|h\|)$. Furthermore, on the event $\big\{\widehat\theta_t\in
K_\epsilon\big\}$,
\begin{equation}
  I(N_t\eta, \widehat{\theta}_t) \le \sup_{\theta \in K_\epsilon}\int \pi(x) \|h(x)\| \mathbf 1\left\{
    \frac{\pi(x)\|h(x)\|}{q(x,\theta)}>N_t\eta \right\}\dd x.
  \label{eq:majore.It}
\end{equation}
Lemma~\ref{lem:A1} implies that the right-hand side of \eqref{eq:majore.It} goes
to $0$ since $N_t\eta$ tends to infinity. Thus,
\begin{equation*}
  \limsup \esp(I(N_t\eta, \widehat\theta_t))  
  \le \pi(\|h\|)\limsup\prob\big(\widehat\theta_t\not\in K_\epsilon\big) \le \pi(\|h\|) \epsilon.
\end{equation*}
Since $\epsilon$ was arbitrary, we have proven that $S_t\to 0$ in $L^1(\prob)$
hence in probability, and condition \textit{(iii)}  is fulfilled.

\paragraph{Conclusion.}
We are now in a position to apply Theorem~\ref{thm:Cappe:1}. Condition
\textit{(i)} is satisfied because, by construction, given $\mathcal F_t$, the
random variables $V_{t,i}$, for $1\le i \le N_t$ are conditionally
independent and $\esp \big[\| V_{t,i}\| \,\big|\, \mathcal F_t\big] <
+\infty$. Condition \textit{(ii)} is easy to check because, here,
\begin{equation*}
  \sum_{i = 1}^{N_t} \esp\Big[\big\|V_{t,i} \,\big\|\,  \Big| \mathcal
  F_t\Big] = \int \|h(x)\| \pi(x) \dd x
\end{equation*}
is not random (thus is a tight sequence). And condition \textit{(iii)} has been
proven above. Hence the conclusion of Theorem~\ref{thm:Cappe:1} is true and
the proof is completed.  \hfill $\qed$

\subsection{Proof of the almost sure convergence of Theorem~\ref{thm:strong}}
\label{sub:prove2}

\paragraph{Bounding the conditional probabilities.} Fix $\eps>0$.
Using that, conditionally on $\widehat\theta_t=\theta$, the $X^t_i$'s
are iid from distribution $Q(\theta)$, we have
\begin{align*}
\var \Big( \big\| \widehat\theta_{t+1} &- \theta^\ast \big\| \Big| 
\widehat{\theta}_t=\theta\Big) 
= \frac{1}{N_t} \int \left\|
   \frac{\pi(x)}{q(x,\theta)}h(x) - \theta ^\ast
 \right\|^2 q(x, \theta) \dd x 
\\
& = \frac{1}{N_t} \int\left\{ \left\|
   \frac{\pi(x)}{q(x,\theta)}h(x)\right\|^2 - 2\left\langle \frac{\pi(x)}{q(x,\theta)}h(x),\theta ^\ast\right\rangle
 +\|\theta^\ast\|^2 \right\}  q(x, \theta) \dd x 
\quad
\\
& =
\frac{v(\theta)-\|\theta^\ast\|^2}{N_t} 
\end{align*}
where $v(\theta)=\int\pi(x)^2\|h(x)\|^2/q(x,\theta)\dd x$ is finite and
continuous because $h\in \mathscr G^2(\mathbb R^d)$. 

With the continuity assumption, $v(\theta)-\|\theta^\ast\|^2$
is bounded from above by some finite constant, $K_\eps$ say, on the
compact ball,
$\bar{\mathcal B}(\theta, \eps)$ say,
centered on $\theta^\ast$, of radius $\eps$ and the conditional
Chebyshev inequality gives, for all $\theta\in \bar{\mathcal B}(\theta^\ast,\eps)$,
\begin{equation}\label{eq:varttt}
  \prob  \Big( \left\| \widehat\theta_{t+1} - \theta^\ast \right\| > \eps\Big| 
\widehat{\theta}_t=\theta\Big)\le \frac{K_\eps}{\eps^2N_t}.
\end{equation}
Multiplying by $\mathbf 1\{\theta\in \bar{\mathcal B}(\theta^\ast,\eps)\}$
on both side of the above inequality and integrating over the
distribution of $\widehat{\theta}_t$ leads to
\begin{equation}
  \label{eq:chebyshev}
   \prob  \Big( \left\| \widehat\theta_{t+1} - \theta^\ast \right\| > \eps\Big| 
   \|\widehat{\theta}_t-\theta^\ast\|\le \eps\Big)\le \frac{K_\eps}{\eps^2N_t}.
\end{equation}

\paragraph{Proving the almost sure convergence.}
Now, we recall that $\widehat\theta_t$ forms a (time- inhomogeneous)
Markov chain. Thus, using \eqref{eq:chebyshev},
\begin{align*}
\prob\bigg(
\bigcap_{t=T}^{T'} \|\widehat\theta_{t+1} - \theta^\ast\| \le \eps
\bigg) & = \prob\Big(\|\widehat\theta_{T+1} - \theta^\ast\|\le\eps\Big)
\prod_{t=T+1}^{T'-1} \prob\bigg(
\| \widehat\theta_{t+1} - \theta^\ast\| \le \eps \bigg|
\|\widehat\theta_{t}-\theta^\ast\|\le\eps
\bigg)
\\
&\ge \prob\Big(\|\widehat\theta_{T+1}-\theta^\ast\|\le\eps\Big)
\prod_{t=T+1}^{T'-1} \left( 1 - \frac{K_\eps}{\eps^2N_t} \right).
\end{align*}
And, when $T'\to\infty$, we obtain
\[
\prob\bigg(
\bigcap_{t\ge T} \|\widehat\theta_{t+1}-\theta^\ast\|\le\eps
\bigg) \ge \prob\Big(\|\widehat\theta_{T+1}-\theta^\ast\|\le\eps\Big)
\prod_{t\ge T+1} \left( 1 - \frac{K_\eps}{\eps^2N_t} \right).
\]
Applying the logarithm on the product and classical results on series,
because $\sum_t 1/N_t$ is finite, the infinite product $\prod_{t}( 1 -
{K_\eps}/{\eps^2N_t})$ converges (that is to say the limit is finite
and strictly positive). In particular, the remainder of the infinite
product in the right hand side of the above inequality tends to $1$
when $T\to\infty$. Furthermore, because of the convergence in
probability proven in Paragraph~\ref{sub:prove2},
$\prob\Big(\|\widehat\theta_{T+1}-\theta^\ast\|\le\eps\Big)$ tends also
to $1$ and thus
\[
\lim_{T\to\infty} \prob\bigg(
\bigcap_{t\ge T} \|\widehat\theta_{t+1}-\theta^\ast\|\le\eps
\bigg) =1.
\]
And we have proved that $\limsup_{T\to\infty}
\|\widehat\theta_{T}-\theta^\ast\| \le\eps$ almost surely. Since
$\eps$ is arbitrary, this proves the desired almost sure convergence.

\section{Proof of the convergence of the final recycling scheme}
\label{sec:proofs2}

The last step is to prove Theorem~\ref{thm:Consistency},
\textit{i.e.}, that the AMIS estimator 
\begin{equation*}
  \widehat{\Pi}^{\text{AMIS}}_T(\psi) = \frac{1}{\Omega_T} \sum_{t = 1}^T
  \sum_{i = 1}^{N_t} \frac{\pi\big(X_i^t\big)}{\Omega_T^{-1}\sum_{k = 1}^T
           N_kq\Big(X_i^t, \widehat{\theta}_k\Big)} \psi\big(X_i^t\big),  
\end{equation*}
which is the result of the unique recycling step of our scheme, is 
consistent for the integral $\Pi(\psi)$ for a large class of functions
$\psi$.
To this aim, we set 
\begin{equation}\label{eq:auxi}
  \widehat{\Pi}^\ast_T(\psi) := \frac{1}{\Omega_T} \sum_{t = 1}^T
  \sum_{i = 1}^{N_t} \frac{\pi(X_i^t)}{q\big(X_i^t, \theta^\ast\big)}\psi(X_i^t),  
\end{equation}
(which cannot be computed in practice because $\theta^\ast$ is unknown).
Note that the auxiliary variable defined in \eqref{eq:auxi} is a
(weighted) average of the random variables
$\widehat{\pi}_t^\ast(\psi)$ given by 
\begin{equation*}
  \widehat{\pi}_t^\ast(\psi) := \frac{1}{N_t}\sum_{i = 1}^{N_t} 
  \frac{\pi(X_i^t) }{q\big(X_i^t, \theta^\ast\big)} \psi(X_i^t).
\end{equation*}
We also define
\begin{equation}
  \label{eq:Iast}
  I^\ast_\psi(\theta) := \int \left[\frac{\pi(x)\psi(x)}{q(x,
    \theta^\ast)} \right] q(x,\theta) \dd x
\end{equation}
which is the conditional expectation of $\pi_t^\ast(\psi)$ knowing
that $\widehat\theta_t=\theta$.

The proof is organised as follows. After stating useful lemmas, we
prove that the sequence of auxiliary variables are strongly
consistent. We then show that the difference between our estimator and
this auxiliary variable, namely $ \widehat{\Pi}^{\text{AMIS}}_T(\psi)
-\widehat{\Pi}^\ast_T(\psi)$, tends to $0$ almost surely. Hence the
consistency stated in Theorem~\ref{thm:Consistency}.

\subsection{Technical results on the function of class $\mathscr H^2(\mathbb R)$}

The proof of the following results, which deal with the continuity
condition imposed in the definition of $\mathscr G^2(\mathbb R)$, is
left to reader, be very similar to the proof of
Lemma~\ref{lem:G2}. Likewise, this lemma implies that, if some
function $\psi$ belongs to $\mathscr H^2(\mathbb R)$, and if $\varphi$
is another function dominated by $\psi$, then $\varphi \in \mathscr
H^2(\mathbb R)$. 
\begin{lem}\label{lem:H2}
  Assume that, for any $\theta \in \Theta$, the integral
  \[
  w_\psi(\theta) := \int
  {\pi^2(x) \psi^2(x)} \Big/ {q^2\big(x,\theta^\ast\big)}q(x,\theta) \dd x
  \]
  is finite. These conditions are equivalent:
  {\em (i)} $w_\psi$ is continuous at $\theta^\ast$; and 
  {\em (ii)} when $\theta\to\theta^\ast$,
    \[
    \int \frac{\pi^2(x) \psi^2(x)}{q^2\big(x,\theta^\ast\big)} \big|
    q(x,\theta) - q(x,\theta^\ast) \big| \dd x \to 0.
    \]
\end{lem}

We also need the following to control the conditional expected values
of $\pi_t^\ast(\psi)$.
\begin{lem} \label{lem:Iast} If, $\psi\in \mathscr H^2(\mathbb R)$,
  then the integrals $I_\psi^\ast(\theta)$ defined in \eqref{eq:Iast}
  are well defined for all $\theta$ and the map $I_\psi^\ast$ is
  continuous at $\theta=\theta^\ast$.
\end{lem}
\begin{proof} Fix any $\theta\in \Theta$.
  If $X_\theta$ is a
  random variable with distribution $Q(\theta)$, then
  $Y_\theta=\pi(X_\theta) \psi(X_\theta)/q(X_\theta, \theta^\ast)$ is
  square integrable since $\psi$ is in $\mathscr H^2(\mathbb
  R)$. Hence $Y_\theta$ is a $L^1$-random variable and its expected
  value, namely $I_\psi(\theta)$ is well defined.

  Now, set $g(x)=\pi(x) \psi(x) / q(x, \theta^\ast)$. We have
  $|g(x)|\le \max(1, g^2(x))$, thus
  \[
  \int |g(x)|\big| q(x,\theta) - q(x,\theta^\ast) \big| \dd x \le
  \int \big| q(x,\theta) - q(x,\theta^\ast) \big| \dd x
  +
  \int g^2(x) \big| q(x,\theta) - q(x,\theta^\ast) \big| \dd x
  \]
  The first integral in this bound goes to $0$ because of Scheff\'e's
  Theorem, see e.g. \citet{Billingsley}, Theorem 16.12 p. 215. The second integral
  goes also to $0$, because $\psi$ is in $\mathscr H^2(\mathbb R)$ and
  because of Lemma~\ref{lem:H2}. Whence
  \[
  |I_\psi^\ast(\theta)-I_\psi^\ast(\theta^\ast)| \le
  \int |g(x)| \big| q(x,\theta) - q(x,\theta^\ast) \big| \dd x \to 0. \qedhere
  \]
\end{proof}

\subsection{Convergence of some auxiliary variables}

\begin{pro}\label{pro:petit.pi}
  Assume that $h\in\mathscr G^2(\mathbb R^d)$, $\sum_t 1/N_t$ is
  finite and  $\psi\in \mathscr H^2(\mathbb R)$. 
  When $t\to\infty$,
  $\big(\widehat{\pi}_t^\ast(\psi) -
  I^\ast_\psi(\widehat{\theta}_{t-1}) \big)$ tends to $0$ almost
  surely. 
  Moreover, under those assumptions,
  $I_\psi^\ast(\widehat\theta_t)\to I_\psi^\ast(\theta^\ast) = \Pi(\psi)$. 
\end{pro}
\begin{proof} The last result follows from
  Theorem~\ref{thm:strong} and continuity of $I_\psi^\ast$ at
  $\theta^\ast$ proven in Lemma~\ref{lem:Iast}.  Adapting the proof
  written in Paragraph~\ref{sub:prove1}, we also have that
  $\big(\widehat{\pi}_t^\ast(\psi) -
  I^\ast_\psi(\widehat{\theta}_{t-1}) \big)$ tends to $0$ in
  probability
 
  The rest of the proof is inspired from the part of the proof of
  Theorem~\ref{thm:strong} written in Paragraph~\ref{sub:prove2}. 
  Fix $\eps>0$. Set $\Delta_{t+1}=\widehat{\pi}_{t+1}^\ast(\psi) -
    I^\ast_\psi(\widehat{\theta}_{t})$, 
    $A_{t+1}=\Big\{ |\Delta_{t+1}|\le \eps\Big\} \cap \Big\{
    \|\widehat\theta_{t+1}-\theta^\ast\|\le \eps \Big\}$ and $\bar
    A_{t+1}$ the complementary event.
  We have
  \[ 
  \esp\Big\{ (\Delta_{t+1})^2\Big|
  \widehat\theta_t=\theta\Big\} = \frac{1}{N_t} \bigg[\int
  \Big(\pi(x) \psi(x)\big/ q\big(x,
  \theta^\ast\big)\Big)^2 q\big(x, \theta\big) \dd x -
  I^\ast(\theta)^2 \bigg].
  \]
  and the term between brackets is bounded from above by some
  finite $K_\eps'$ for any $\theta\in \bar{\mathcal
  B}(\theta^\ast,\eps)$ because $\psi\in \mathscr H^2(\mathbb R)$. 
  Thus, for all $\theta\in \bar{\mathcal B}(\theta^\ast,\eps)$,
  \begin{align*}
  \prob\Big(\bar A_{t+1} \Big| \widehat\theta_t=\theta
  \Big) &\le
  \prob\Big(|\Delta_{t+1}| > \eps\ \Big| \widehat\theta_t=\theta \Big) + \prob\Big(
  \|\widehat\theta_{t+1}-\theta^\ast\|\ge\eps \Big| \widehat\theta_t=\theta
  \Big)
  \\ 
  & \le \frac{K_\eps'+K_\eps}{\eps^2N_t},
  \end{align*}
  where we have used the Chebyshev inequality and resorted
  \eqref{eq:varttt} to bound the first and second
  probabilities respectively.
  Hence $\prob \Big(\bar A_{t+1} \Big| A_t
  \Big)\le K/N_t$ for some finite $K$.
  Now, since $(\Delta_{t},\widehat{\theta}_t)$ is a
  (time-inhomogeneous) Markov chain, we have
  \[
  \prob\left(\bigcap_{t=T}^\infty  A_{t+1} \right) \ge 
  \prob \left( A_{T+1} \right)
  \prod_{T=t}^\infty \left(1-\frac{K}{N_t}\right)
  \]
  The above infinite product tends to $1$ since $\sum_t 1/N_t$ is
  finite.
  And $\prob \left( A_{T+1} \right)\to 1$ because both $\Delta_{T+1}$
  and $\|\widehat{\theta}_{T+1}-\theta^\ast\|$ tends to $0$ in probability.
\end{proof}

We shall now recall that the auxiliary variable $
\widehat{\Pi}_T^\ast(\psi)$ is a weighted average of the
$\widehat{\pi}_t^\ast(\psi)$ for $t=1,\ldots, T$ which are controlled
by Proposition~\ref{pro:petit.pi} proven above.  The following Lemma
is obvious, using Ces\`aro Lemma on sequence of (non random) vectors.
\begin{lem} \label{lem:cesaro} %
  Let $\{U_t\}$ be a sequence of random vectors and $U$ another
  random vector. If $\{b_t\}$ is a sequence of positive
  real numbers such that $B_t=b_1 + \ldots + b_t$ tends to infinity,
  then the event $\Big\{ U_t\to U \Big\}$ is included in the event
  $\Big\{ B_t^{-1}\sum_{k = 1}^t b_k U_k \to U_\infty\Big\}$.
\end{lem}
This Ces\`aro Lemma and Proposition~\ref{pro:petit.pi} above leads to the following.
\begin{pro} \label{pro:auxiliary} Assume that $h\in\mathscr
  G^2(\mathbb R^d)$, $\sum_t 1/N_t$ is finite and $\psi\in \mathscr
  H^2(\mathbb R)$.  Then
  \begin{equation*}
    \widehat{\Pi}_T^\ast(\psi) \longrightarrow  \Pi(\psi) \quad \text{almost surely}.
  \end{equation*}
\end{pro}
\begin{proof} Because of 
  Proposition~\ref{pro:petit.pi}, $\widehat{\pi}_t^\ast(\psi)$ tends
  almost surely to $\Pi(\psi)=\int\psi(x)\pi(x)\dd x$.
  Applying Lemma~\ref{lem:cesaro} with
  $b_t=N_t$ yields the convergence of $\widehat{\Pi}_T^\ast(\psi)$.
\end{proof}

\subsection{Controlling the discrepency between the AMIS estimator and the
  auxiliary variable}

The convergence of $\widehat{\Pi}_T^{AMIS}(\psi) -
\widehat{\Pi}_T^\ast(\psi)$ towards $0$ almost surely is proven in
Proposition~\ref{pro:AMIS} below, whose proof relies on some preliminary
result given in Lemma~\ref{lem:D}.  To this end, we define the function
$D_T\big(\cdot \big) : \mathscr X \mapsto \mathbb R_+$
by
\[
D_T\big(x \big) = 
\Omega_T^{-1}\sum_{k = 1}^T N_k q\left(x, \widehat \theta_k \right) 
\]
which appears in the denominator of the updated
weight \eqref{eq:weight:t}. Because of the consistency of the learning
scheme proven in Section~\ref{sec:proofs}, we are able to show in the
following lemma that this denominator resembles the denominator of the
classical importance weight, when the proposal distribution is $Q(\theta^\ast)$.
\begin{lem}
 \label{lem:D}
 Let $K$ be a compact subset of $\Theta$.  The event
 $\Big\{\widehat\theta_t\to\theta^\ast \Big\}$ is included in the
 event where
 \begin{equation*}
   \lim_{T \to +\infty} \left\|
     \frac{q(\cdot,\theta^\ast)}{D_T(\cdot)}-1
     \right\|_{K,\infty}
   = 0. 
 \end{equation*}
\end{lem}
\begin{proof}
  Denote by $m_{\eps,K}$ the infimum of $m_\eps(x)$ on $K$, where
  $m_\eps(\cdot)$ is the function defined in~\eqref{eq:mesp}.
  Actually, $m_{\eps,K}$ is the infimum of the lower semicontinuous function $(x,\theta)
  \mapsto q(x,\theta)$ on the compact set $K\times \bar{\mathcal
    B}(\theta^\ast,\eps)$. Since a lower semicontinuous function
  attains its lower bound on any compact set, and $q(x,\theta)>0$ for
  all $x$ and $\theta$, the infimum $m_{\eps, K}$ is positive.

  Now fix a point of the probability space in the event
  $\Big\{\widehat\theta_t\to\theta^\ast \Big\}$. There, there
  exists some $t_\eps$ such that, for all $t>t_\eps$,
  $\|\widehat\theta_t - \theta^\ast\|<\eps$. Hence, for all
  $T>t_\eps$, and all $x\in \mathscr X$,
  \[
  D_T(x)\ge \frac 1{\Omega_T}\sum_{k=t_\eps+1}^T N_k q(x,\theta_k)
  \ge \frac{\Omega_T-\Omega_\eps}{\Omega_T}m_\eps(x)
  \quad
  \text{where }\Omega_\eps=\sum_{k=1}^{t_\eps}N_k.
  \]
  Therefore  
  \begin{align}
    \left| \frac{q(x,\theta^\ast)}{D_T(x)}-1
    \right| & \le \frac{\Omega_T}{(\Omega_T-\Omega_\eps)m_\eps(x)} 
    \left| q(x,\theta^\ast) - {D_T(x)}
    \right|\nonumber
    \\
    & \le \frac{\Omega_T}{(\Omega_T-\Omega_\eps)m_{\eps,K}}\Omega_T^{-1}\sum_{t = 1}^T N_k
    \left\| q(\cdot,\theta^\ast) - {q(\cdot,\widehat\theta_t)}
    \right\|_{K,\infty}\label{eq:DolceGabana}.
  \end{align}
  The bound in \eqref{eq:DolceGabana} is uniform on $K$ and goes to
  $0$ using Lemma~\ref{lem:cesaro}, which leads to the desired result.
\end{proof}


We can now state and prove the result controlling the difference between the AMIS
estimator and the auxiliary variable $\widehat{\Pi}_T^\ast(\psi)$.
\begin{pro} \label{pro:AMIS} %
  Assume that $h\in \mathscr G^2(\mathbb R^d)$ and $\sum_t
  1/N_t<\infty$. Moreover, assume that, for some $\eps>0$,
  $\psi(\cdot) q(\cdot,\theta^\ast)/m_{\eps}(\cdot)$ is in $\mathscr
  H^2(\mathbb R)$. Then
  \begin{equation*}
    \lim_{T \to +\infty}\widehat{\Pi}_T^{\text{AMIS}}(\psi)
    - \widehat{\Pi}_T^\ast(\psi) = 0 \quad \text{almost surely}.  
  \end{equation*}
\end{pro}

\begin{proof} Fix $\alpha>0$.
  The integral
  \[
  \int\Big|\psi(x)\Big| \frac{q(x,\theta^\ast)}{m_\eps(x)} \pi(x)\dd x
  \]
  is finite because $|\psi(\cdot)| q^\ast(\cdot) / m_\eps(\cdot) \in
  \mathscr H^2(\mathbb R)$.  Therefore we can find some compact subset
  $K$ of $\mathscr X$ such that
  \[
  \int_{\mathscr X\setminus K}\Big|\psi(x)\Big|
  \frac{q(x,\theta^\ast)}{m_\eps(x)} \pi(x)\dd x < \alpha.
  \]
  Now, set $\psi_1(x) := \psi(x)\mathbf 1\{x\in K\}$, 
  $\psi_2(x) := \psi(x)\mathbf 1\{x \not\in K\}$ so that
  $\psi(x)=\psi_1(x)+\psi_2(x)$. And consider the event 
  \[
  E=\Big\{ \widehat\theta_t\to\theta^\ast \Big\} \cap 
  \Big\{ \piast_T(|\psi_1|)\to \Pi(|\psi_1|) \Big\} 
  \cap 
  \Big\{ \piast_T(|\psi_2|)\to \Pi(|\psi_2|) \Big\}
  \cap 
  \left\{ \piast_T\left( \varphi \right)
    \to \Pi\left( \varphi\right)
    \right\}.
  \]
  where $\varphi(x):= |\psi_2(x)| {q(x,\theta^\ast)} \Big/ {m_\eps(x)}$.
  With Theorem~\ref{thm:strong} and Proposition~\ref{pro:auxiliary}, this event
  is of probability $1$. Moreover, note that, because of
  \eqref{eq:mesp}, $q(x,\theta^\ast)/m_\eps(x)\ge 1$ and thus
  \begin{equation}
    \label{eq:alpha}
    \Pi(|\psi_2|)\le\Pi(\varphi) = \int_{\mathscr X \setminus K}|\psi_2(x)|
    \frac{q(x,\theta^\ast)}{m_\eps(x)}\pi(x)\dd x <\alpha.
  \end{equation}

  Then, using
  linearity of the operators $\piamis_T$ and $\piast_T$, we have
  \begin{align}
    \left| \piamis_T(\psi)-\piast_T(\psi) \right| & \le
    \left| \piamis_T(\psi_1)-\piast_T(\psi_1) \right| +
    \left| \piamis_T(\psi_2)-\piast_T(\psi_2) \right| \notag
    \\
    & \le
    \left| \piamis_T(\psi_1)-\piast_T(\psi_1) \right| +
    \piamis_T(|\psi_2|) + \piast_T(|\psi_2|).
    \label{eq:3bounds}
  \end{align}
 The first term in the right hand side of \eqref{eq:3bounds} can be
 controlled as follow:
  \begin{align*}
   \Delta_T:= \left| \piamis_T(\psi_1)-\piast_T(\psi_1) \right|  
    & \le \frac{1}{\Omega_T}\sum_{t=1}^T\sum_{i=1}^{N_t} \frac{
      \pi(X^t_i) \Big|\psi_1(X^t_i)\Big|}{q(X^t_i,\theta^\ast)} 
    \left\| \frac{q(\cdot,\theta^\ast)}{D_T(\cdot)} - 1 \right\|_{K,\infty}
    \\
    & \le \left\| \frac{q(\cdot,\theta^\ast)}{D_T(\cdot)} - 1 \right\|_{K,\infty} \piast_T(|\psi_1|).
  \end{align*}
  On the event $E$, using Lemma~\ref{lem:D}, the first term of the
  last bound goes to $0$, and $\piast_T(|\psi_1|)\to\Pi(|\psi_1|)$.
  Hence $\lim_T \Delta_T = 0$ on $E$.

  On the event $E$, the second term of the right hand side of
  \eqref{eq:3bounds} can be bounded by
  \begin{align}
    \piamis_T(|\psi_2|) & \le \frac 1{\Omega_T}
    \sum_{t=1}^T\sum_{i=1}^{N_t}
    \frac{\pi(X^t_i) |\psi_2(X^t_i)|}{q(X^t_i,\theta^\ast)}\
    \frac{q(X^t_i, \theta^\ast)}{D_T(X^t_i)} \notag
    \\
    &\le \frac{\Omega_T}{\Omega_T-\Omega_\eps}
    \piast_T\left( |\psi_2(\cdot)|
      \frac{q(\cdot,\theta^\ast)}{m_\eps(\cdot)}\right)
    =\frac{\Omega_T}{\Omega_T-\Omega_\eps}\piast_T(\varphi)
    \label{eq:2nd}
  \end{align}
  using the fact that, on $E$, there
  exists a $t_\eps$ such that, for all $t>t_\eps$,
  $\|\widehat\theta_t - \theta^\ast\|<\eps$. Hence, for all
  $T>t_\eps$, and all $x\in \mathscr X$,
  \[
  D_T(x)\ge \frac 1{\Omega_T}\sum_{k=t_\eps+1}^T N_k q(x,\theta_k)
  \ge \frac{\Omega_T-\Omega_\eps}{\Omega_T}m_\eps(x)
  \quad
  \text{where }\Omega_\eps=\sum_{k=1}^{t_\eps}N_k.
  \]
  On the event $E$, $\piast_t(\varphi)$ converges to $\Pi(\varphi)$
  which is smaller than $\alpha$ because of
  \eqref{eq:alpha}. Moreover, $(\Omega_T - \Omega_\eps) / \Omega_T \to
  1$. Hence, on the event $E$,
  \[
  \limsup_T \piamis_T(|\psi_2|) \le \alpha
  \]
  
  And, finally, on the event $E$, the third term of the right hand
  side of \eqref{eq:3bounds} converges to $\Pi(|\psi_2|)$ which is
  smaller than $\alpha$ using \eqref{eq:alpha}. Hence, on $E$,
  \[
  \limsup_T \piast_T(|\psi_2|) \le \alpha
  \]

  Reporting in \eqref{eq:3bounds}, we
  obtain that, on the event $E$ of probability $1$,
  \[
  \limsup_T \left| \piamis_T(\psi)-\piast_T(\psi) \right|
  \le 2 \alpha.
  \]
  Because $\alpha$ is arbitrary small, we have proven the desired result.
\end{proof}

\subsection{Conclusion  of the proof of Theorem~\ref{thm:Consistency}}
Proposition~\ref{pro:auxiliary} gives the convergence of the auxiliary
variable $\widehat{\Pi}^\ast_T(\psi)$ towards the integral $\Pi(\psi)$ almost
surely, while the discrepency between the AMIS estimator and this
auxiliary variable becomes negligeable almost surely
(Propositon~\ref{pro:AMIS}).  Whence the
almost sure convergence of the estimator 
$\widehat{\Pi}^{\text{AMIS}}_T(\psi)$  to the integral $\Pi(\psi)$, and then 
the proof of Theorem~\ref{thm:Consistency} is completed.
\hfill $\qed$

\section{Numerical experiments}
\label{sec:numerical}

We provide here a detailed numerical example on which we have compared
\begin{enumerate}
\item[(a)] an iterative and adaptive algorithm learning $\theta$ with a
  naive recycling strategy at the end;
\item[(b)] the original AMIS of \citet{cornuet:2012} and
\item[(c)] our modified AMIS with its recycling step only at the end.
\end{enumerate}
We refer the reader to \citet{cornuet:2012} for the orginal
algorithm~(b). The algorithm~(c) is the main topic of this paper, and is
described with great details in Section~\ref{sec:algorithms}. At
last, algorithm~(a) follows the same code lines, but stops at
line~\ref{line:end.learn} and returns a merging of all past samples 
without updating the weights computed at line~\ref{line:simple.weight}.
In this example, the target is the posterior distribution
when conducting a Bayesian analysis on a population genetic data set.
It turns out that our algorithm was the most powerful
for a given computational cost, that is to say, for a given number of
simulations from proposals.

\begin{figure}[htb]
  \includegraphics[width = 1\textwidth, height=.6\textwidth]{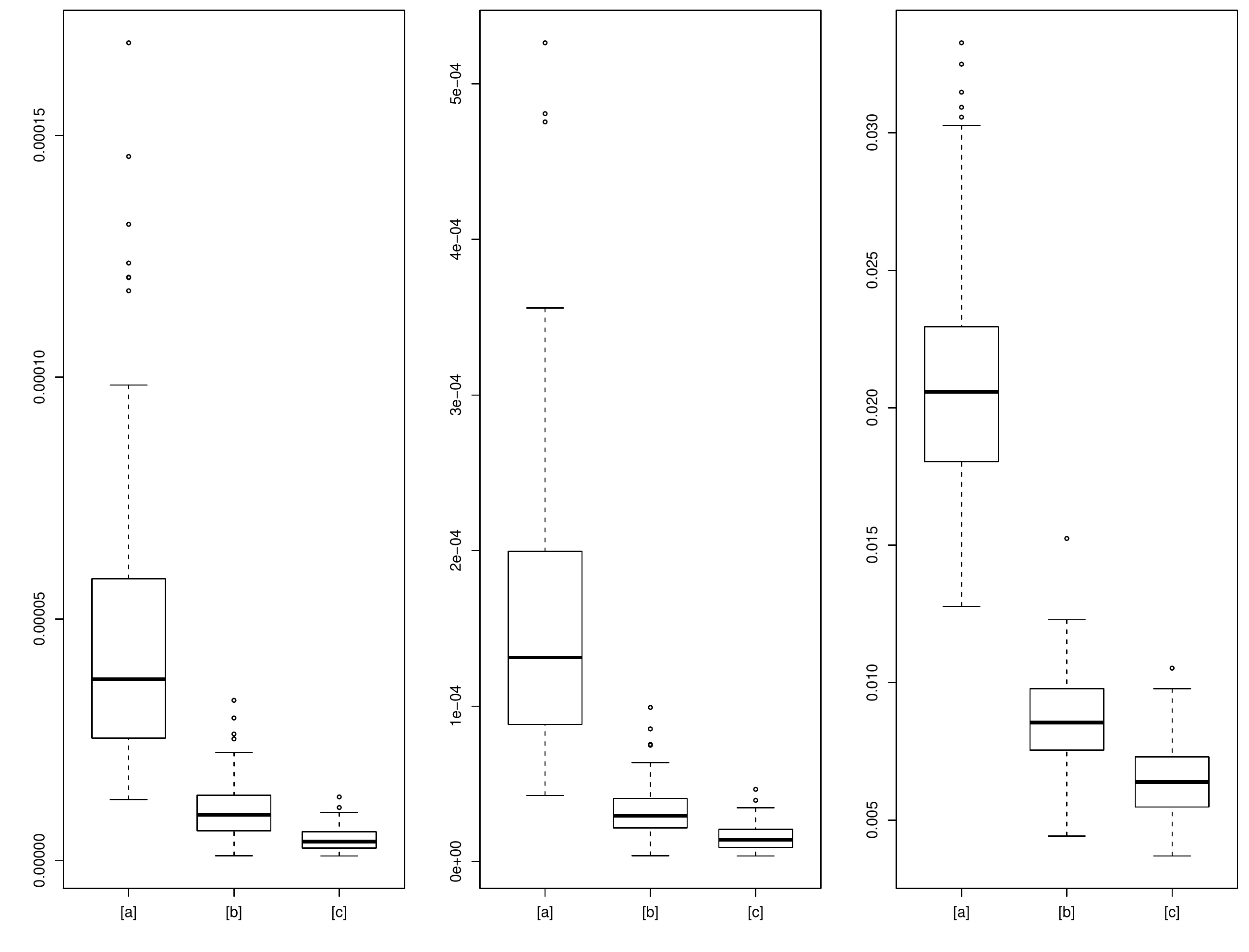}
  \caption{\small Comparison of the three algorithms (a), (b) and (c)
    over $100$ replicates based on distances between the discrete
    probability function given by the output and the target in the
    population genetic example. \textit{Left:} Cramer-von-Mises
    distances. \textit{Middle:} $\mathbb
    L^2$-distances. \textit{Right:} $\mathbb L^\infty$-distances.}
\label{fig:GenPop}
\end{figure}

\paragraph{A Bayesian model in population genetics.}
This example comes from a population genetic problem. More
precisely, we want to conduct a Bayesian analysis of a genetic data set
$\mathscr D$ to infer mutation and migration rates in a parametric
model. Assume that the species of interest is composed of two large
populations at equilibrium, one on an island and the other one on the
mainland. The parametric model we have used is coalescent based, and
is detailed, for instance in \citet{donnelly:tavare:1995} and
\citet{rousset:2012}. The genetic data come from two samples of
individuals corresponding respectively to the two populations,
genotyped at five independent microsatellite loci.

The model is composed of two populations, whose effective population
sizes are both equal to 10000. We restrict the demographic scenario of
this model to a symmetric migration between the two populations, and
the migration rates are the supposed to be the same in both
directions. We consider the mutation model SMM (\textit{Single
  Mutation Model}). Our data set which we denote by $\mathscr D$
is simulated on five independent loci. At each locus we simulate the
genotypes of individuals, using the software \textit{IBDSim} of
\citet{leblois:2009}. For this data set, we set the mutation rate
$x_{\text{mut}}$ to $2.3$ and the symmetric migration rate
$x_{\text{mig}}$ to $0.04$. In this example, the
likelihood of a data set $\ell\left(\mathscr D | x_{\text{mut}},
  x_{\text{mig}}\right)$ is the product of five integrals.  Each
integral represents the likelihood of the data set at a given
locus. In this study, we approximate these integrals through
importance sampling methods. These approximations are provided by the
software \textit{Migraine} of \citet{rousset:2012}.

We consider a uniform prior on the set $\mathscr X = (10^{-1}, 10) \times
(10^{-3}, 0.5)$, thus simplifying the expression of the posterior
density to
\begin{equation}
  \label{eq:posterior:GenPop}
  \pi\big(x_{\text{mut}}, x_{\text{mig}} | \mathscr
  D\big) \propto \ell\big(\mathscr D
  | x_{\text{mut}}, x_{\text{mig}}\big) \mathbf 1_{\mathscr X}\big(x_{\text{mut}}, x_{\text{mig}}\big).
\end{equation}

Hence the target $\pi(x)$ is the posterior distribution on a two dimensional
parameter $x=(x_\text{mig}, x_\text{mut})$ when the prior is a non informative
uniform distribution on some  set $\mathscr X$. This example is actually
typical of situations where the density of the target $\pi(x)$ is of high
computational cost. With coalescent based models, the likelihood, thus the
posterior density $\pi(x)$, is an integral over a latent process that, by
chance, is computed via importance sampling too, see \citet{deIorio:2005} and
\citet{rousset:2012}.

\paragraph{Tuning of the algorithms.}
The family of proposals in the three sequential algorithms is composed
of bivariate Gaussian distributions, conditioned (or truncated) on the
support $\mathscr X$ of the prior distribution. The parameter of the proposal is a
four dimensional vector $\theta=(\mu_\text{mig}, \mu_\text{mut},
\sigma^2_\text{mig}, \sigma^2_\text{mut})$,  whose  first two coordinates
give the position of the mode and last two coordinates give the
marginal variances of the diagonal covariance matrix.
For each realization 
of a scheme, we set $T = 45$ and $N_t = N \times t$ where $N = 100$.

\paragraph{Results.}
Performances of the sampling algorithms were compared as follows. Instead of
looking at the estimates of $\Pi(\psi)$ for various integrands $\psi$, we
decided here to evaluate the outputs with distances between the discrete
measure induced by the weighted final samples and the target. In
Figure~\ref{fig:GenPop}, we have represented the Cramer-von Mises, the
$\mathbb L^2$- and the $\mathbb L^\infty$- distances between the empirical
distribution function $\widehat{F}_T(x)$ of the final weighted samples and the
distribution function $F(x)$ of the target, \textit{i.e.}, the posterior
distribution. Furthermore, since the density of the target cannot be written
with a close formula in the concrete example, we shall describe how $F(x)$ was
computed. Following an idea of R. Leblois and F. Rousset, estimates of
$\pi(x)$ were computed for values of $x$ ranging a regular $500\times 500$
grid of the support of the prior distribution. The estimation error was then
decreased using a kriging model on $\pi(x)$, assuming regularity conditions of
that posterior density. Of course, this sharp approximation comes at a much
higher computational cost than any run of the Monte Carlo algorithms we
compare here.

\bigskip

The results presented in Figure~\ref{fig:GenPop} exhibit a clear
advantage to our modified AMIS, namely (c), in front of (a),
the sequential scheme with a naive recycling and (b), the original
AMIS, whatever the distance. Thus, modifying the original AMIS was not
only a way to obtain the theoretical results of
Section~\ref{sec:results}, but also a real improvement of the original
algorithm. One of the reason that might explain this phenomenon is
that the recycling scheme of the original AMIS introduces a bias on
$\theta$ during the learning process which tends to accumulate, and
thus is large enough to degrade the output quality when compared to
our modified AMIS.

\section{Conclusion and discussion}
\label{sec:discussion}
For a certain class of functions, we derived strong consistency of
our modified AMIS. We proved a strong law of large numbers for a
large class of integrands characterized by regularity conditions and for a
general family of proposals. We assumed that the size of the samples
at each stage, namely $N_t$, tends to infinity rather quickly so that
$\sum_t 1/N_t$ is finite. This condition might be unsatisfactory but is due
to the fact that we only assumed that $\pi(X)\|h(X)\|/q(X,\theta)$ has
a finite quadratic moment when $X\sim Q(\theta)$.  In future research, we could try to
relaxe the hypotheses on $N_t$, assuming that the above random
variable has exponential moments and using a large deviation
inequalities instead of the Chebyshev bound.
%

Besides, another route might be taken to prove theoretical
results on the modified AMIS, based on Markovian arguments. Indeed,
when the sample size does not vary between iterations during the
learning process, \textit{i.e.}, when $N_1=N_2=\ldots=N_T=N$, the
sequence of pairs $(X^t_{1:N}, \widehat\theta_t)$ form a Markov
chain. And the final sum in \eqref{eq:amis:estimator} might be traited
using results on averages over a path of a Markov chain. But we have
left this route for future works.

The present paper does not conduct a comprehensive numerical
comparison between the original AMIS and our algorithm, because we
focused here on theoretical results. But the numerical experiment
presented above is a serious example exhibiting an advantage to our
modification of the AMIS. 

Finally two important, methodological issues have not been tackle in
this theoritical work. 
The first one deals with the initialization of the AMIS. The original
paper of \citet{cornuet:2012} proposed an answer based on a logistic
sample when nothing is known on the target. We stress here that the
starting distribution is of great practical consequence: for instance,
if the first sample misses a mode of the target distribution, we have
almost no chance to see it during the whole process. That was summed
up  by \citet{cornuet:2012} as the ``what-you-get-is-what-you-see''
nature of the AMIS.
Likewise, a recurring numerical question on the AMIS concerns the
allocation of the overall computational cost (given by the final system size
$\Omega_{T}$).  To optimize allocation, one could propose and study
allocation strategies on different iterations via the sequence $N_1,\ldots,
N_T$. We also believe that the winner in the competition between the
original and the modified AMIS depends also on this allocation strategy

\appendix

\section*{Appendix -- Weak law of large numbers on triangular arrays }
\label{AnnexeTabTriang}

The following result is a multidimensional generalisation of the weak
law of large numbers given in Chapter 9 of \cite{cappe:book}.  In the
following, all random variables are assumed to be defined on a joint
probability space $\big(\Xi, \mathcal F, \mathbb P \big)$ and
$\big\{N_t\big\}_{t \ge 1}$ denotes an increasing sequence of
integers.  In those theorem as well as throughout this paper, we used
the notion of tightness of random variables that we recall here. (See
\citet{Billingsley}, p. 336 for more details)

\begin{defi}
  A sequence of random vectors $\{U_n\}$ is tight if
  \begin{equation*}
    \lim_{\eta \to \infty}\ \sup_{n\ge 1} \prob \Big(\big\|U_n\big\| \ge \eta \Big) = 0.
  \end{equation*}
\end{defi}

\bigskip

We have the following law of large numbers.
\begin{theo}
  \label{thm:Cappe:1}
  Let $\big\{V_{t,i} \big\}_{1 \le i \le N_t}$ be a triangular array
  of random vectors on $\R^d$, and let $\big\{ \mathcal F_t \big\}_{t \ge
    1}$ be a sequence of $\sigma$-fields. Assume that the following
  conditions hold true.
  \begin{enumerate}
  \item[\it (i)] The $V_{t,i}$ for $i=1,\ldots, N_t$ are conditionally independent given
    $\mathcal F_t$,  and for any $t$ and $i = 1, \ldots, N_t$,
    $\esp \bigg[\big\|V_{t,i}\big\| \Big|
    \mathcal F_t\bigg] < +\infty$.
  \item[\it (ii)] The sequence $\bigg\{ \sum_{i=1}^{N_t} \esp\Big[\big\|V_{t,i}\big\| 
    \Big| \mathcal F_t\Big]\bigg\}_{t \ge 1}$ is tight.
  \item[\it (iii)] For any positive $\eta$,
    \begin{equation*}
      \lim_{t\to+\infty}\sum_{i=1}^{N_t} \esp\bigg[\big\|V_{t,i}\big\|\mathbf 
      1\bigg\{\big\|V_{t,i}\big\| > \eta \bigg\} \Big| \mathcal F_t \bigg] 
      = 0 \quad \text{in probability.}
    \end{equation*}
  \end{enumerate}
  Then,
  \begin{equation*}
    \lim_{t\to\infty}\sum_{i=1}^{N_t} \bigg\{V_{t,i} - \esp\Big[V_{t,i} 
      \Big| \mathcal F_t \Big]\bigg\}
    =0\quad \text{in probability.}
  \end{equation*}
\end{theo}

\section*{Acknowledgments} This research was financially supported
by the French ``\textit{A\-gen\-ce Nationale de la Recherche}''
through the ANR project EMILE (ANR-09-BLAN-0145-04). All three authors
thank Christian P. Robert for useful discussions on the AMIS. The
first two author were also financially supported by the \textit{Labex}
NUMEV (\textit{Solutions Num\'eriques,Mat\'erielles et Mod\'elisation
  pour l'Envi\-ron\-nement et le Vivant}). The
last author is grateful to Rapha\"el Leblois for his help in
population genetics.


\bibliographystyle{apalike}
\bibliography{AMISBib}
\end{document}